\documentclass[a4paper,11pt]{article}
\pdfoutput=1
\usepackage[margin=1in]{geometry}
\usepackage{amsmath}
\usepackage{amsfonts}
\usepackage{amssymb}
\usepackage{stmaryrd}
\usepackage{pst-all}
\usepackage{times}
\usepackage[T1]{fontenc}

\vfuzz \hfuzz
\colorlet{darkred}{red!77!black}

\newcommand\buchi{B\"{u}chi}

\newcommand\fnht{\mathsf{fnht}}

\newcommand\infi{\mathsf{infin}}
\newcommand\mft{\mathsf{mft}}
\newcommand\nc{\mathsf{nc}}
\renewcommand\next{\mathsf{next}}
\newcommand\opt{\mathsf{opt}}

\newcommand\pri{\mathsf{pri}}

\newcommand\s{\mathfrak s}
\newcommand\T{\mathcal T}

\newtheorem{theorem}{Theorem}[section]
\newtheorem{lemma}[theorem]{Lemma} 
 
\newtheorem{corollary}[theorem]{Corollary}

\newcommand\qed{\hfill $\Box$}

\newenvironment{proof}{\noindent\textbf{Proof. }}{\nopagebreak
 \qed\medskip}

\newenvironment{proofidea}{\noindent\textit{Proof idea. }}{\nopagebreak
  \medskip}

\begin{document}
\title{Tight Bounds for Complementing Parity Automata}

\author{Sven Schewe and Thomas Varghese \\ Department of Computer Science, University of Liverpool}
\date{}

\bibliographystyle{plain}
\pagestyle{plain}

\maketitle
\begin{abstract}
We follow a connection between tight determinisation and complementation and establish a complementation procedure from transition-labelled parity automata to transition-labelled nondeterministic B\"uchi automata. We prove it to be tight up to an $O(n)$ factor, where $n$ is the size of the nondeterministic parity automaton. This factor does not depend on the number of priorities.
\end{abstract}

\section{Introduction}
The precise complexity of complementing $\omega$-automata is an intriguing problem for two reasons:
first, the quest for optimal algorithms is a much researched problem~\cite{Buchi/62/Automata,Sakoda+Sipser/78/finTight,Pecuchet/86/Complement,Sistla+Vardi+Wolper/87/Complement,Safra/88/Safra,Michel/88/lowerComplementation,Thomas/99/complementation,Loding/99/omega,Kupferman+Vardi/01/Weak,GKSV/03/Complementing,FKV/06/tighter,Piterman/07/Parity,Vardi/07/Saga,Yan/08/lowerComplexity}, and second, complementation is a valuable tool in formal verification (c.f.,~\cite{kurshan/94/verification}), in particular when studying language inclusion problems of $\omega$-regular languages.
Complementation is also useful to check the correctness of translation techniques~\cite{Vardi/07/Saga,DBLP:conf/cav/TsaiTH13}.
The GOAL tool~\cite{DBLP:conf/cav/TsaiTH13}, for example, provides such a test suite and incorporates recent algorithms~\cite{Safra/88/Safra,Thomas/99/complementation,Kupferman+Vardi/01/Weak,Piterman/07/Parity} for B\"uchi complementation.

While devising optimal complementation algorithms for nondeterministic finite automata is simple---%
nondeterministic \emph{finite} automata can be determinised using a simple subset construction, and deterministic finite automata can be complemented by complementing the set of final states~\cite{Rabin+Scott/59/finite,Sakoda+Sipser/78/finTight}---%
devising optimal complementation algorithms for nondeterministic $\omega$-automata is hard, because simple subset constructions are not sufficient to determinise or complement them~\cite{Michel/88/lowerComplementation,Loding/99/omega}.

Given the hardness and importance of the problem, the complementation of \mbox{$\omega$-automata} enjoyed much attention. The initial focus was on the complementation of B\"uchi automata with state-based acceptance \cite{Buchi/62/Automata,Pecuchet/86/Complement,Sistla+Vardi+Wolper/87/Complement,Michel/88/lowerComplementation,Safra/88/Safra,Loding/99/omega,Thomas/99/complementation,Kupferman+Vardi/01/Weak,GKSV/03/Complementing,FKV/06/tighter,Vardi/07/Saga,Yan/08/lowerComplexity,Schewe/09/complementation,DBLP:conf/cav/TsaiTH13}, and it resulted in a continuous improvement of its upper and lower bounds.

The first complementation algorithm dates back to the introduction of \buchi\ automata in 1962.
In his seminal paper ``On a decision method in restricted second order arithmetic'' \cite{Buchi/62/Automata}, \buchi\ develops a doubly exponential complementation procedure.
While \buchi's result shows that nondeterministic \buchi\ automata (and thus \mbox{$\omega$-regular} expressions) are closed under complementation, complementing an automaton with $n$ states may, when using \buchi's complementation procedure, result in an automaton with $2^{2^{O(n)}}$ states, while an $\Omega(2^n)$ lower bound~\cite{Sakoda+Sipser/78/finTight} is inherited from finite automata.

In the late 80s, these bounds have been improved in a first sequence of results, starting with establishing an EXPTIME upper bound~\cite{Pecuchet/86/Complement,Sistla+Vardi+Wolper/87/Complement}, which matches the EXPTIME lower bound~\cite{Sakoda+Sipser/78/finTight} inherited from finite automata.
However, the early EXPTIME complementation techniques produce automata with up to $2^{O(n^2)}$ states~\cite{Pecuchet/86/Complement,Sistla+Vardi+Wolper/87/Complement};
hence, these upper bounds were still exponential in the lower bounds.

This situation changed in 1988, when Safra introduced his famous determinisation procedure for nondeterministic \buchi\ automata~\cite{Safra/88/Safra}, resulting in an $n^{O(n)}$ bound for \buchi\ complementation, while Michel~\cite{Michel/88/lowerComplementation} established a seemingly matching $\Omega(n!)$ lower bound in the same year.
Together, these results imply that \buchi\ complementation is in $n^{\theta(n)}$, leaving again the impression of a tight bound.

As pointed out by Vardi~\cite{Vardi/07/Saga}, this impression is misleading, because the $O()$ notation hides an $n^{\theta(n)}$ gap between both bounds.
This gap has been narrowed down in 2001 to $2^{\theta(n)}$ by the introduction of an alternative complementation technique that builds on level rankings and a cut-point construction~\cite{Kupferman+Vardi/01/Weak}.
The complexity of the plain method is approximately $(6n)^n$~\cite{Kupferman+Vardi/01/Weak}, leaving a $(6e)^n$ gap to Michel's lower bound~\cite{Michel/88/lowerComplementation}.

Subseqently, tight level rankings~\cite{FKV/06/tighter,Yan/08/lowerComplexity}
have been exploited by Friedgut, Kupferman, and Vardi~\cite{FKV/06/tighter} to improve the upper complexity bound to $O\big((0.96 n)^n\big)$, and by Yan~\cite{Yan/08/lowerComplexity} to improve the lower complexity bound to $\Omega\big((0.76 n)^n\big)$.
Schewe \cite{Schewe/09/complementation} has provided a matching upper bound, showing tightness up to an $O(n^2)$ factor.

In recent works, more succinct acceptance mechanism have been studied, where the most important ones are parity and generalised B\"uchi automata, as they occur naturally in the translation of $\mu$-calculi and LTL specifications, respectively. 
In \cite{Schewe+Varghese/12/generalisedBuchi}, we gave tight bounds for the determinisation and complementation of generalised \buchi\ automata.
For Rabin, Streett, and parity automata, there has been much progress \cite{DBLP:conf/lics/CaiZL09,DBLP:conf/csl/CaiZ11,DBLP:conf/fsttcs/CaiZ11}, in particular establishing an
$n^{\theta(n)}$ bound for parity complementation with state-based acceptance, which has been a great improvement and pushed tightness of parity~comple-  mentation to the level known from \buchi\ complementation since the late 80s \cite{Safra/88/Safra,Michel/88/lowerComplementation}.

\noindent{\bf Contribution.}\hspace{1ex}In this paper, we establish tight bounds for the complementation of parity automata with transition-based acceptance.
A generalisation of the ranking-based complementation procedures quoted above to transition-based acceptance is straight forward, and the Safra-style determinisation procedures from the literature \cite{Safra/88/Safra,Safra/92/Streett,Piterman/07/Parity,Schewe/09/determinise,Schewe+Varghese/12/generalisedBuchi} have a natural representation with an acceptance condition on transitions.
Their translation to state-based acceptance is by multiplying the acceptance from the last transition to the state space.

A similar observation can be made for other automata transformations, like the removal of $\varepsilon$-transitions from translations of $\mu$-calculi \cite{Wilke/01/Alternating,Schewe+Finkbeiner/06/ATM} and the treatment of asynchronous systems \cite{Schewe+Finkbeiner/06/Asynchronous}, where the state-space grows by multiplication with the acceptance information (e.g., maximal priority on a finite sequence of transitions), while it cannot grow in case of transition-based acceptance.
Similarly, tools like \textsc{SPOT} \cite{Duret-Lutz/11/spot} offer more concise automata with transition-based acceptance mechanism as a translation from LTL.
Using state-based acceptance in the automaton that we want to complement would also complicate the presentation of the complementation procedure.
But first and foremost, using transition-based acceptance provides cleaner results.

This is the case because in state-based acceptance, the role of the states is overloaded.
In finite automata over infinite structures, each state represents the class of tails of the word that can be accepted from this state.
In state-based acceptance, they have to account for the acceptance mechanism itself, too, while they are relieved from this burden in transition-based acceptance.
In complementation techniques based on rankings, this results in a situation where states with certain properties, such as final states for B\"uchi automata, can only occur with some ranks, but not with all.

As transition-based acceptance separates these concerns, the presentation becomes cleaner.
The natural downside is that we lose the $n^{O(n)}$ bound \cite{DBLP:conf/csl/CaiZ11} for parity complementation, as the number of priorities in a parity automaton with transition-based acceptance can grow arbitrarily.
But in return, we do get a clean and simple complementation procedure based on a data structure we call flattened nested history trees (FNHTs), which is inspired by a generalisation of history trees \cite{Schewe/09/determinise} to multiple levels, one for each even priority $\geq 2$.

In \cite{Schewe+Varghese/12/generalisedBuchi}, we showed a connection between optimal determinisation and complementation for generalised B\"uchi automata, where we exploit the nondeterministic power of a B\"uchi automaton to devise a tight complementation procedure. In this paper, we follow this connection between tight determinisation~\cite{Schewe+Varghese/14/parityDeterminisation} and complementation to devise a tight complementation construction from parity to nondeterministic B\"uchi automata.

We show that any procedure that complements full parity automata with states $Q$ and maximal priority $\pi$ has at least $|\fnht(Q,\pi)|/2$ states, where $\fnht(Q,\pi)$ is the set of FNHTs for a given set $Q$ of states and maximal priority $\pi$ of the parity automaton that is to be complemented.
Our complementation construction uses a marker in addition for its acceptance mechanism.
Essentially, it is used to mark some position of interest in an FNHT.
It accounts for the $O(n)$ gap between the upper and lower bound.
We show that, for $\pi \geq 2$ (and hence for B\"uchi automata upwards) the number of states of our complementation construction is bounded by $4n+1$ times the lower bound.

\section{Preliminaries}
We denote the non-negative integers by $\omega = \{0,1,2,3,...\}$.
For a finite alphabet $\Sigma$, an infinite \emph{word} $\alpha$ is an infinite sequence $\alpha_0 \alpha_1 \alpha_2 \cdots$ of letters from $\Sigma$.
We sometimes interpret $\omega$-words as functions $\alpha : i \mapsto \alpha_i$, and use $\Sigma^{\omega}$ to denote the 
$\omega$-words over~$\Sigma$.

$\omega$-automata are finite automata that are interpreted over infinite words and recognise $\omega$-regular languages $L\subseteq \Sigma^{\omega}$.
Nondeterministic parity automata are quintuples \mbox{$\mathcal P=(Q,\Sigma,I,T,\pri:T \rightarrow \Pi)$}, where $Q$ is a finite set of states with a non-empty subset $I\subseteq Q$ of initial states, $\Sigma$ is a finite alphabet, $T\subseteq Q \times \Sigma \times Q$ is a transition relation that maps states and input letters to sets of successor states, and $\pri$ is a priority function that maps transitions to a finite set $\Pi \subset \omega$ of non-negative integers.

A \emph{run} $\rho$ of a nondeterministic parity automaton $\mathcal P$ on an input word $\alpha$ is an infinite sequence $\rho: \omega\rightarrow Q$ of states of $\mathcal P$, also denoted $\rho = q_0 q_1 q_2\cdots \in Q^{\omega}$, such that the first symbol of $\rho$ is an initial state $q_0\in I$ and, for all $i\in\omega$, $(q_{i},\alpha_{i},q_{i+1})\in T$ is a valid transition.
For a run $\rho$ on a word $\alpha$, we denote with $\overline{\rho}: i \mapsto \big(\rho(i),\alpha(i),\rho(i+1)\big)$ the transitions of $\rho$.
Let $\infi(\rho) = \{q\in Q \mid \forall i \in \omega \; \exists j>i\mbox{ such that } \rho(j)=q\}$ denote the set of all states that occur infinitely often during the run $\rho$.
Likewise, let $\infi(\overline{\rho})=\{t\in T \mid \forall i \in \omega \; \exists j>i\mbox{ such that } \overline{\rho}(j)=t\}$ denote the set of all transitions that are taken infinitely many times in $\overline{\rho}$.
Acceptance of a run is defined through the priority function $\pri$.
A run $\rho$ of a parity automaton is \emph{accepting} if $\limsup_{n\rightarrow\infty} \pri\big(\overline{\rho}(n)\big)$ is even, that is, if the highest priority that occurs infinitely often in the transitions of $\rho$ is even.
A word $\alpha$ is accepted by a parity automaton $\mathcal P$ iff it has an accepting run, and its language $\mathcal L(\mathcal P)$ is the set of words it accepts.

Parity automata with $\Pi \subseteq \{1,2\}$ are called \emph{B\"uchi} automata.
B\"uchi automata are denoted $\mathcal B=(Q,\Sigma,I,T,F)$, where $F \subseteq T$ are called the final or accepting transitions.
A run is accepting if it contains infinitely many accepting transitions.
$\mathcal B$ is thus a rendering of the parity automaton, where $\pri:t \mapsto 2$ if $t \in F$ and $\pri:t \mapsto 1$ if $t \notin F$.

We assume w.l.o.g.\ that the set $\Pi$ of priorities satisfies that $\min \Pi \in \{0,1\}$.
If this is not the case, we can simply change $\pri$ accordingly to $\pri':t \mapsto \pri(t) - 2$ several times until this constraint is satisfied.
We likewise assume that $\Pi$ has no holes, that is, $\Pi = \{ i \in \omega \mid \max \Pi \geq i \geq \min \Pi\}$.
If there is a hole $h \notin \Pi$ with $\max \Pi > h > \min \Pi$, we can change $\pri$ to $\pri' : t \mapsto \pri(t)$ if $\pri(t) < h$ and $\pri' : t \mapsto \pri(t)-2$ if $\pri(t) > h$.
Obviously, these changes do not affect the acceptance of any run, and applying finitely many of these changes brings $\Pi$ into this normal form.

The different priorities have a natural order $\succcurlyeq$, where $i \succ j$ if $i$ is even and $j$ is odd; $i$ is even and $i > j$; or $j$ is odd and $i < j$.
For a non-empty set $\Pi' \subseteq \Pi$ of priorities, $\opt \Pi' = \{i \in \Pi' \mid \forall j \in \Pi'.\ i \succcurlyeq j\}$ denotes the optimal priority for acceptance.

The complexity of a parity automaton $\mathcal P=(Q,\Sigma,I,T,\pri:T \rightarrow \Pi)$ is measured by its size $n = |Q|$ and its set of priorities $\Pi$.
For a given size $n$ and set of priorities $\Pi$, there is an automaton that recognises a hardest language.
This automaton is referred to as the full automaton $\mathcal P_n^\Pi=(Q,\Sigma,I,T,\pri:T \rightarrow \Pi)$, with
$|Q| = n$, $I=Q$, $\Sigma= Q \times Q \rightarrow 2^\Pi$, $T = \{q,\sigma,q') \mid \sigma(q,q') \neq \emptyset$, and $\pri(q,\sigma,q') = \opt \sigma(q,q')$.

Note that partial functions from $Q \times Q$ to $\Pi$ would work as well as the alphabet.
The larger alphabet is chosen for technical convenience in the proofs.
Any other language recognised by a nondeterministic parity automaton $\mathcal P$ with $n$ states and priorities $\Pi$ can essentially be obtained by a language restriction via alphabet restriction from $\mathcal P_n^\Pi$.

\section{Complementing parity automata}
The construction described in this section draws from two main sources of inspiration.
One source is the introduction of efficient techniques for the determinisation of parity automata in \cite{Schewe+Varghese/14/parityDeterminisation}.
The nested history trees used there have been our inspiration for the \emph{flattened nested history trees} that form the core data structure in the complementation from Subsection \ref{ssec:complement} and are the backbone of the lower bound proof from Subsection \ref{ssec:lower}.
The second source of inspiration is the connection \cite{Schewe+Varghese/12/generalisedBuchi} between the efficient determinisation based on history trees \cite{Schewe/09/determinise} for B\"uchi automata and generalised B\"uchi automata \cite{Schewe+Varghese/12/generalisedBuchi} and their level ranking based complementation \cite{Kupferman+Vardi/01/Weak,FKV/06/tighter,Schewe/09/complementation,Schewe+Varghese/12/generalisedBuchi}.

The intuition for the complementation is to use the nondeterministic power of a B\"uchi automaton to reduce the size of the data stored for determinisation.
As usual, this nondeterministic power is intuitively used to guess a point in time, where all nodes of the nested history trees from parity determinisation \cite{Schewe+Varghese/14/parityDeterminisation}, which are eventually always stable, are henceforth stable.
Alongside, the set of stable nodes can be guessed.

Like in the construction for B\"uchi automata, the structure can then be flattened, preserving the `nicking order', the order in which older nodes and descendants take preference in taking states of the nondeterministic parity automaton that is determinised.
The complement automaton runs in two phases: a first phase before this guessed point in time, and a second phase after this point, where the run starts in such a flattened tree.

In the first subsection, we introduce \emph{flattened nested history trees} as our main data structure.
While we take inspiration from nested history trees \cite{Schewe+Varghese/14/parityDeterminisation}, the construction is self-contained.
In the second subsection, we show that B\"uchi automata recognising the complement language of the full nondeterministic parity automaton $\mathcal P_n^\Pi$ need to be large by showing disjointness properties of accepting runs for a large class of words, one for each full flattened nested history tree introduced in Subsection \ref{ssec:fnht}.
The definition of this language is also instructive in how the data structure is exploited.

We extend our data structure by markers, resulting in \emph{marked flattened trees}, which are then used as the main part of the state space of the natural complementation construction introduced in Subsection \ref{ssec:complement}.
We show correctness of our complementation construction in Subsection \ref{ssec:correct} and tightness up to an $O(n)$ factor in Subsection \ref{ssec:tight}.

Note that all our constructions assume $\max \Pi \geq 2$, and therefore do not cover the less expressive CoB\"uchi automata.

\subsection{Flattened nested history trees \& marked flattened trees}
\label{ssec:fnht}
Flattened nested history trees (FNHTs) are the main data structure used in our complementation algorithm.
For a given parity automaton $\mathcal P=(Q,\Sigma,I,T,\pri:T \rightarrow \Pi)$, an FNHT over the set $Q$ of states, maximal priority $\pi_m = \max \Pi$ and maximal even priority  $\pi_e = \opt \Pi$, is a tuple $(\T,l_s:\T \rightarrow 2^Q,l_l:\T \rightarrow 2\mathbb N,l_p:\T \rightarrow 2^Q,l_r:\T \rightarrow 2^Q)$, where $\T$(an ordered, labelled tree) is a non-empty, finite, and prefix closed subset of finite sequences of natural numbers and a special symbol $\mathfrak s$ (for \emph{stepchild}), $\omega \cup \{\s\}$, that satisfies the constraints given below.
We call a node $v \s \in \T$ a \emph{stepchild} of $v$, and refer to all other nodes $vc$ with $c\in \omega$ as the \emph{natural children} of $v$.
$\nc(v) = \{ vc \mid c\in \omega$ and $vc \in \T\}$ is the set of natural children of $v$.
The root is a stepchild.

The \textbf{constraints} an FNHT quintuple has to satisfy are as follows:
\begin{itemize}
 \item Stepchildren have only natural children, and natural children % have
 only stepchildren.
 \item Only natural children and, when the highest priority $\pi$ is odd, the root may be leafs.
 \item $\T$ is order closed: for all $c,c'\in \omega$ with $c<c'$, $vc' \in \T$ implies $vc \in \T$.
 \item For all $v \in \T$, $l_s(v) \neq \emptyset$.
 \item If $v$ is a stepchild, then $l_p(v) = \emptyset$.
 \item If $v$ is a stepchild, then $l_s(v) = l_r(v) \cup \bigcup_{v' \in \nc(v)} l_s(v')$.
 
The sets $l_s(v')$ and $l_s(v'')$ are disjoint for all $v',v'' \in \nc(v)$ with $v' \neq v''$, and $l_r(v)$ is disjoint with $\bigcup_{v' \in \nc(v)} l_s(v')$.
 \item If $v$ is a natural child, then $l_p(v) \neq \emptyset$, $l_s(v) = l_p(v) \cup l_r(v)$, and $l_p(v) \cap l_r(v) = \emptyset$. % is disjoint with $l_r(v)$.
 \item If a natural child $v$ is not a leaf, then $l_s(v\s)=l_p(v)$.
 \item $l_l(\varepsilon)=\pi_e$ and, for all $v \in \T$, $l_l(v)\geq 2$.
 \item If $v\s \in \T$, then $l_l(v\s) = l_l(v) - 2$, and if $vc \in \T$ for $c \in \omega$, then $l_l(vc) = l_l(v)$.
\end{itemize}

The elements in $l_s(v)$ are called the states, $l_p(v)$ the pure states, and $l_r(v)$ the recurrent states of a node $v$, and $l_l(v)$ is called its level.
Note that the level follows a simple pattern:
the root is labelled with the maximal even priority, $l_l(\varepsilon)=\pi_e$,
the level of natural children is the same as the level of their parents, and the level of a stepchild $v\s$ of a node $v$ is two less than the level of $v$.
For a given maximal even priority $\pi_e$, the level is therefore redundant information that can be reconstructed from the node and $\pi_i$.
For a given set $Q$ and maximal priority $\pi$, $\fnht(Q,\pi)$ denotes the flattened nested history trees over $Q$.
An FNHT is called \emph{full} if the states $l_s(\varepsilon)=Q$ of the root is the full set $Q$.

To include an acceptance mechanism, we enrich FNHTs to \emph{marked flattened tress} (MFTs), which additionally contain a \emph{marker} $v_m$ and a marking set $Q_m$, such that
\begin{itemize}
 \item either $v_m = (\overline{v},r)$ with $\overline{v} \in \T$ is used to mark that we follow a breakpoint construction on the recurrent states, in this case $l_r(\overline{v}) \supseteq Q_m \neq \emptyset$,
 \item or $v_m = (\overline{v},p)$ such that $\overline{v}$ is a leaf in $\T$ is used to mark that we follow a breakpoint construction on the pure states of a leaf $\overline{v}$, in this case $l_p(\overline{v}) \supseteq Q_m \neq \emptyset$. 
\end{itemize}

The marker is used to mark a property to be checked.
For markers $v_m = (\overline{v},r)$, the property is that a particular node would not spawn stable children in a nested history tree \cite{Schewe+Varghese/14/parityDeterminisation}.
As usual in Safra like constructions, this is checked with a breakpoint, where a breakpoint is reached when all children of a node spawned prior to the last breakpoint die.
For markers $v_m = (\overline{v},p)$, the property is that all runs that are henceforth trapped in the pure nodes of $v$ must eventually encounter a priority $l_l(v)-1$. This priority is then dominating, and implies rejection as an odd priority.
We check these properties round robin for all nodes in $\T$, skipping over nodes, where the respective sets $l_r(\overline{v})$ or $l_p(\overline{v})$ are empty, as the breakpoint there is trivially reached immediately.
 
For a given FNHT $(\T,l_s,l_l,l_p,l_r)$, $\next(v_m)$ is a mapping from a marker $v_m$ to a marker/marking set pair $(\overline{v},r),l_r(\overline{v})$ or $(\overline{v},p),l_p(\overline{v})$. The new marker is the first marker after $v_m$ in some round robin order such that the set $l_r(\overline{v})$ or $l_p(\overline{v})$, resp., is non-empty.

If $(\T,l_s,l_l,l_p,l_r)$ is an FNHT and $v_m$ and $Q_m$ satisfy the constraints for markers and marking sets from above, then $(\T,l_s,l_l,l_p,l_r;v_m,Q_m)$ is a marked flattened tree.
For a given set $Q$ and priorities $\Pi$ with maximal priority $\pi = \max \Pi$, $\mft(Q,\pi)$ denotes the marked flattened trees over $Q$.
A marking is called \emph{full} if either $v_m = (\overline{v},r)$ and $Q_m = l_r(\overline{v})$, or $v_m = (\overline{v},p)$ and $Q_m = l_p(\overline{v})$.

\subsection{Construction}
\label{ssec:complement}

For a given nondeterministic parity automaton $\mathcal P = (Q,\Sigma,I,T,\pri:T \rightarrow \Pi)$ with maximal even priority $\pi_e >1$, we construct a nondeterministic B\"uchi automaton $\mathcal C = (Q',\Sigma,\{I\},T',F)$ that recognises the complement language of $\mathcal P$ as follows.
First we set $Q'= Q_1 \cup Q_2$ with $Q_1 = 2^Q$ and $Q_2 = \mft(Q,\pi)$, and $T' = T_1 \cup T_t \cup T_2$, where
\begin{itemize}
 \item $T_1 \subseteq Q_1 \times \Sigma \times Q_1$ are transitions in an initial part $Q_1$ of the states of $\mathcal C$,
 \item $T_t \subseteq Q_1 \times \Sigma \times Q_2$ are transfer transitions that can be taken only once in a run, and
 \item $T_2 \subseteq Q_2 \times \Sigma \times Q_2$, are transitions in a final part $Q_2$ of the states of $\mathcal C$,
\end{itemize}
where $T_1$ and $T_2$ are deterministic.
We first define a transition function $\delta$ for the subset construction and functions $\delta_i$ for all priorities $i \in \Pi$, and then the sets $T_1$, $T_t$, and $T_2$.

\begin{itemize}
 \item $\delta : (S,\sigma) \mapsto \{q \in Q \mid \exists s \in S.\ (s,\sigma,q) \in T\}$,
 
 \item $\delta_i : (S,\sigma) \mapsto \big\{q \in Q \mid \exists s \in S.\ (s,\sigma,q) \in T$ and $\pri\big(\big(s,\sigma,q)\big) \succcurlyeq i\big\}$,

\item $T_1 = \big\{ (S,\sigma,S') \in Q_1 \times \Sigma \times Q_1 \mid S' = \delta(S,\sigma)\big\}$,

where only transitions $(\emptyset,\sigma,\emptyset)$ are accepting.

\item $T_t = \big\{ \big(S,\sigma,(\T,l_s,l_l,l_p,l_r;v_m,Q_m) \big) \in Q_1 \times \Sigma \times Q_2 \mid l_s(\varepsilon) = \delta(S,\sigma)\big\}$ and we have that$(\T,l_s,l_l,l_p,l_r;v_m,Q_m)$ is a marked flattened tree.

\item $T_2 = \big\{ \big((\T,l_s,l_l,l_p,l_r;v_m,Q_m),\sigma,s \big) \in Q_2 \times \Sigma \times Q_2 \mid$
\begin{itemize}
 \item if $v$ is a stepchild, then $l_s''(v) = \delta_{l_l(v)+1}\big(l_s(v),\sigma\big)$
 \item if $v$ is a natural child, then $l_s''(v) = \delta_{l_l(v)-1}\big(l_s(v),\sigma\big)$
 \item if $v$ is a natural child, then $l_r''(v) = \delta_{l_l(v)-1}\big(l_r(v),\sigma\big) \cup \delta_{l_l(v)}\big(l_s(v),\sigma\big)$,

 \item starting at the root, we then define inductively:
\begin{itemize}
 \item $l_s'(\varepsilon) = l_s''(\varepsilon)$, 
  \item if $vc$ is a natural child, then $l_s'(vc) = \big(l_s''(vc) \cap l_s'(v)\big) \smallsetminus \bigcup_{c' < c} l_s''(vc')$,
        $l_r'(vc)= l_r''(vc) \cap l_s'(vc)$, and $l_p'(vc)= l_s'(vc) \smallsetminus l_r'(vc)$, and
 \item if $v\s$ is a stepchild, then $l_s'(v\s) = l_p'(v)$.
\end{itemize}

 \item if one exists, we extend the functions to obtain the unique FNHT $(\T,l_s',l_l,l_p',l_r')$ \hfill (otherwise $\mathcal C$ blocks) 
 %\pagebreak
 
 \item if $v_m = (\overline{v},r)$ then $Q_m' = \delta_{l_l(\overline{v})-1}(Q_m,\sigma) \cap l_r'(\overline{v})$, and

       if $v_m = (\overline{v},p)$ then $Q_m' = \delta_{l_l(\overline{v})-3}(Q_m,\sigma) \cap l_p'(\overline{v})$,
       
 \item if $Q_m' = \emptyset$, then the transition is accepting and $s = \big(\T,l_s',l_l,l_p',l_r';\next(v_m)\big)$,
 
 \item if $Q_m' \neq \emptyset$, then the transition is not accepting and $s = (\T,l_s',l_l,l_p',l_r';v_m,Q_m')$.
\end{itemize}
\end{itemize}

\subsection{Correctness}
\label{ssec:correct}
To show that $\mathcal L(\mathcal C)$ is the complement of $\mathcal L(\mathcal P)$, we first show that a word accepted by $\mathcal C$ is rejected by $\mathcal P$ and then, vice versa, that a word accepted by $\mathcal P$ is rejected by $\mathcal C$.

\begin{lemma}
If $\mathcal C$ has an accepting run on $\alpha$, then $\mathcal P$ rejects $\alpha$.
\end{lemma}

\begin{proof}
Let $\rho = S_0 S_1 \ldots$ be an accepting run of $\mathcal C$ on $\alpha$ that stays in $Q_1$.
Thus, there is an $i \in \omega$ such that, for all $j\geq i$, $S_j = \emptyset$.
But if we consider any run $\rho' = q_0 q_1 q_2 \ldots$ of $\mathcal P$ on $\alpha$, then it is easy to show by induction that
$q_k \in S_k$ holds for all $k \in \omega$, which contradicts $S_i = \emptyset$; that is, in this case $\mathcal P$ has no run on $\alpha$.

Let us now assume that $\rho = S_0 S_1 \ldots S_i s_{i+1} s_{i+2} \ldots$ is an accepting run of $\mathcal C$ on $\alpha$, where $(S_i,\alpha_i,s_{i+1}) \in T_t$ is the transfer transition taken.
(Recall that runs of $\mathcal C$ must either stay in $Q_1$ or contain exactly one transfer transition.)

Let us assume for contradiction that $\mathcal P$ has an accepting run $\rho' = q_0 q_1 q_2 \ldots$ with even dominating priority $e = \limsup_{j\rightarrow \infty} \pri\big((q_j,\alpha_j,q_{j+1})\big)$.
Let, for all $j > i$, $s_j = (\T,l_s^j,l_l,l_p^j,l_r^j;v_m^j,Q_m^j\big)$ and $S_j = l_s^j(\varepsilon)$.
It is again easy to show by induction that $q_j \in S_j$ for all $j\in \omega$.
Let now $v_j \in \T$ be the longest node with $l_l^j(v_j) \geq e$ and $q_j \in l_s^j(v_j)$.
Note that such a node exists, as $q_j \in S_j = l_s^j(\varepsilon)$ holds.
We now distinguish the two cases that the $v_j$ do and do not stabilise eventually.
\smallskip

\noindent{\bf First case.}\hspace{1ex} Assume that there are an $i'>i$ and a $v \in \T$ such that, for all $j\geq i'$, $v_j = v$.
We choose $i'$ big enough that $\pri(q_{j-1},\alpha_{j-1},q_j) \succ e + 1$ holds for all $j \geq i'$.

\textbf{If $v$ is a stepchild}, then $q_j \in l_r^j(v)$ for all $j\geq i'$.
Using the assumption that $\rho$ is accepting, there is an $i''>i'$ such that $(s_{i'' - 1},\alpha_{i'' - 1},s_{i''})$ is accepting, and $v_m^{i''} = (v,r)$.
(Note that $q_{i''} \in l_r^{i''}(v)$ implies $l_r^{i''}(v) \neq \emptyset$.)
But then we have $q_{i''} \in Q_m^{i''}= l_r^{i''}(v)$, and an inductive argument provides $(s_j,\alpha_j,s_{j+1})\notin F$ and $q_j \in Q_m^j$ for all $j\geq i''$.  
This contradicts that $\rho$ is accepting.

\textbf{If $v$ is a natural child}, then we distinguish three cases. The first one is that there is a $j' \geq i'$ such that $q_{j'} \in l_r^{j'}(v)$.
Then we can show by induction that $q_j \in l_r^j(v)$ for all $j\geq j'$ and follow the same argument as for stepchildren, using $i'' > j'$.

The second is that $q_j \in l_p^j(v)$ holds for all $j\geq i'$.
There are now again a few sub-cases that each lead to contradiction.
The first is that $l_l(v) = e$. But in this case, we can choose a $j > i'$ with $\pri\big((q_j,\alpha_j,q_{j+1})\big)=e$ and get $q_{j+1} \in l_r^{j+1}(v)$ (contradiction).
The second is that $l_l(v) > e$ and $v$ is not a leaf.
But in that case, $l_l(v\s) \geq e$ holds and $q_j \in l_p^j(v)$ implies $q_j \in l_p^j(v\s)$, which contradicts the maximality of $v$.
Finally, if $l_l(v) > e$ and $v$ is a leaf of $\T$, we get a similar argument as for stepchildren:
Using the assumption that $\rho$ is accepting, there is an $i''>i'$ such that $(s_{i'' - 1},\alpha_{i'' - 1},s_{i''})$ is accepting, and $v_m^{i''} = (v,p)$.
(Note that $q_{i''} \in l_p^{i''}(v)$ implies $l_p^{i''}(v) \neq \emptyset$.)
But then we have $ q_{i''} \in Q_m^{i''}= l_p^{i''}(v)$, and an inductive argument provides $(s_j,\alpha_j,s_{j+1})\notin F$ and $q_j \in Q_m^j$ for all $j\geq i''$.  
This contradicts that $\rho$ is accepting.
\smallskip

\noindent{\bf Second case.}\hspace{1ex}
Assume that the $v_j$ do not stabilise. Let $v$ be the longest sequence such that $v$ is an initial sequence of almost all $v_j$, and let $i'>i$ be an index such that $v$ is an initial sequence of $v_j$ for all $j \geq i'$. Note that $q_j$ is in $l_s(v_j')$ for all ancestors $v_j'$ of $v_j$.

First, we assume for contradiction that there is a $j > i'$ with $\pri\big((q_j,\alpha_j,q_{j+1})\big) = e' \succ l_l(v)$ (note that the `better than' relation implies that $e' > l_l(v)$ is even).
Then we select a maximal ancestor $v'$ of $v$ with $l_l(v')=e'$; note that such an ancestor is a natural child, as a stepchild has only natural children, and all of them have the same level.

As $v'$ is an ancestor of $v_j$ and $v_{j+1}$, $q_j \in l_s^j(v')$ and $q_{j+1} \in l_s^{j+1}(v')$ hold, and by the transition rules thus imply $q_{j+1} \in l_r^{j+1}(v')$, which contradicts $q_{j+1} \in l_s^{j+1}(v_{j+1})$.
(Note that $l_l(v') > l_l(v) \geq l_l(v_{j+1})$ holds.)

Second, we show that $\pri\big((q_j,\alpha_j,q_{j+1})\big) \preccurlyeq l_l(v) + 1$ holds infinitely many times.
For this, we first note that the non-stability of the sequence of $v_j$-s implies that at least one of the following three events happen for infinitely many $j>i'$.

\begin{enumerate}
\item $v$ is a stepchild, $q_j \in l_s^j(vc)$ for some child $vc$ of $v$, but, for all children $vc'$ of $v$, $q_{j+1} \notin l_s^{j+1}(vc')$,
\item $v$ is a stepchild, $q_j \in l_s^j(vc)$ for some child $vc$ of $v$, and $q_{j+1} \in l_s^{j+1}(vc')$ for some older sibling $vc'$ of $vc$, that is, for $c'>c$, or
\item $v$ is a natural child, $q_j \notin l_s^j(v\s)$, but $q_{j+1} \in l_s^{j+1}(v\s)$.
\end{enumerate}

Note that this is just the counter position to ``$v_j$ stabilises or $v$ is not maximal''.
In all three cases, the definition of $T_2$ requires that $\pri\big((q_j,\alpha_j,q_{j+1})\big) \preccurlyeq l_l(v)+1$.

As the first observation implies that there may only be finitely many transitions with even priority $> l_l(v)$ and the second observation implies that there are infinitely many transitions in $\rho'$ with odd priority $> l_l(v)$, they together imply that $\limsup_{j\rightarrow \infty} \pri\big((q_j,\alpha_j,q_{j+1})\big)$ is odd, which leads to the final contradiction. 
\end{proof}

\begin{lemma}
If $\mathcal P$ has an accepting run on $\alpha$, then $\mathcal C$ rejects $\alpha$.
\end{lemma}

\begin{proof}
Let $\rho = q_0 q_1 q_2 \ldots$ be an accepting run of $\mathcal P$ on $\alpha$ with even dominating priority $e = \limsup_{j\rightarrow \infty} \pri\big((q_j,\alpha_j,q_{j+1})\big)$.

Let us first assume for contradiction that $\mathcal C$ has an accepting run $\rho' = S_0 S_1 \ldots$ which is entirely in $Q_1$.
It is then easy to show by induction that $q_i \in S_i$ holds for all $i \in \omega$, such that no transition of $(S_i,\alpha_i,S_{i+1})$ is accepting.

Let us now assume for contradiction that $\mathcal C$ has an accepting run $\rho' = S_0 S_1 \ldots S_i s_{i+1} s_{i+2} \ldots$, where $(S_i,\alpha_i,s_{i+1}) \in T_t$ is the transfer transition taken.
(Recall that runs of $\mathcal C$ must either stay in $Q_1$ or contain exactly one transfer transition.)
Let further
$s_j = (\T,l_s^j,l_l,l_p^j,l_r^j;v_m^j,Q_m^j\big)$ and $S_j = l_s^j(\varepsilon)$ for all $j > i$.

It is easy to show by induction that, for all $j \in \omega$, $q_j \in S_j$ holds.
We choose an $i_\varepsilon > i$ such that, for all $k\geq i_\varepsilon$, $\pri\big((q_{k-1},\alpha_{k-1},q_k)\big) \leq e$ holds.

Let us now look at the nodes $v \in \T$, such that $q_j \in l_s^j(v)$, where $j \geq i_\varepsilon$.
\smallskip

\noindent{\bf Construction basis.}\hspace{1ex}
We have already shown $q_j \in S_j = l_s^j(\varepsilon)$ for all $j > i$, and thus in particular for all $j \geq i_\varepsilon$.
\smallskip

\noindent{\bf Construction step.}\hspace{1ex}
If, for some \textbf{stepchild} $v\in \T$ with $l_l(v) \geq e$ and some $i_v \geq i_\varepsilon$, it holds for all $j \geq i_v$ that $q_j \in l_s^j(v)$, then the following holds for all $j \geq i_v$:
if $v'\in \nc(v)$ is a natural child of $v$ and $q_j \in l_s^j(v')$, then either $q_{j+1} \in l_s^{j+1}(v')$, or there is a younger sibling $v''$ of $v'$ in $\T$ such that $q_{j+1} \in l_s^{j+1}(v'')$.

As transitions to younger siblings can only occur finitely often without intermediate transitions to older siblings, we have one of the following two cases:
\begin{enumerate}
 \item for all $j \geq i_v$, $q_j \in l_s^j(v)$, but for every natural child $v'$ of $v$, $q_j \notin l_s^j(v')$, or
 \item there is a natural child $v'$ of $v$ and an index $i_{v'} \geq i_v$ such that, for all $j \geq i_{v'}$, $q_j \in l_s^j(v')$.
\end{enumerate}
As $v$ is a stepchild, the first case implies that $q_j \in l_r^j(v)$ for all $j \geq i_v$.
However, using the assumption that $\rho'$ is accepting, there is an $i_v'>i_v$ such that $(s_{i_v' - 1},\alpha_{i_v' - 1},s_{i_v'})$ is accepting, and $v_m^{i_v'} = (v,r)$, as the marker is circulating in a round robin fashion.
(Note that $q_{i_v'} \in l_r^{i_v'}(v)$ implies $l_r^{i_v'}(v) \neq \emptyset$.)
But then we have $q_{i_v'} \in Q_m^{i_v'}= l_r^{i_v'}(v)$, and an inductive argument provides $(s_j,\alpha_j,s_{j+1})\notin F$ and $q_j \in Q_m^j$ for all $j\geq i_v'$. 

In the second case, we continue with $v'$ and the index $i_{v'}$.
\smallskip

If, for some \textbf{natural child} $v\in \T$ with $l_l(v) > e$ and some $i_v \geq i_\varepsilon$, it holds for all $j \geq i_v$ that $q_j \in l_s^j(v)$, then one of the following holds.
\begin{enumerate}
 \item There is an $i_v' \geq i_v$ such that $q_{i_v'} \in l_r^{i_v'}(v)$.
 \item For all $j \geq i_v$, $q_j \in l_p^j(v)$.
\end{enumerate}

In the first case, it is easy to show by induction that $q_j \in l_r^j(v)$ holds for all $j \geq i_{v'}$.
We can then again use the assumption that $\rho'$ is accepting.
Consequently, there is an $i_v''>i_v'$ such that $(s_{i_v'' - 1},\alpha_{i_v'' - 1},s_{i_v''})$ is accepting, and $v_m^{i_v''} = (v,r)$, as the marker is circulating in a round robin fashion.
(Note that $q_{i_v''} \in l_r^{i_v''}(v)$ implies $l_r^{i_v''}(v) \neq \emptyset$.)
But then we have again $q_{i_v''} \in Q_m^{i_v''}= l_r^{i_v''}(v)$, and an inductive argument provides $(s_j,\alpha_j,s_{j+1})\notin F$ and $q_j \in Q_m^j$ for all $j\geq i_v''$. 

In the second case, if $v$ is not a leaf, then it holds for all $j \geq i_{v\s} = i_v$ that $q_j \in l_s^j(v\s)$, and we can continue with $v\s$.
If $v$ is a leaf, we again use the assumption that $\rho'$ is accepting.
Consequently, there is an $i_v'>i_v$ such that $(s_{i_v' - 1},\alpha_{i_v' - 1},s_{i_v'})$ is accepting, and $v_m^{i_v'} = (v,p)$, as the marker is circulating in a round robin fashion.
(Note that $v$ is a leaf and that $q_{i_v'} \in l_p^{i_v'}(v)$ implies $l_p^{i_v'}(v) \neq \emptyset$.)
But then we have $q_{i_v'} \in Q_m^{i_v'}= l_p^{i_v'}(v)$, and an inductive argument provides $(s_j,\alpha_j,s_{j+1})\notin F$ and $q_j \in Q_m^j$ for all $j\geq i_v'$. 
\smallskip

If, for some \textbf{natural child} $v\in \T$ with $l_l(v) = e$ and some $i_v \geq i_\varepsilon$, it holds for all $j \geq i_v$ that $q_j \in l_s^j(v)$, then there is, by the definition of $e$, a $j>i_v$ with $\pri(q_{j-1},\alpha_{j_1},q_j) = e$.
But then $q_{j-1} \in l_s^{j-1}(v)$ and $q_j \in l_s^j(v)$ imply $q_j \in l_r^j(v)$.
It is then easy to establish by induction that $q_{j'} \in l_r^{j'}(v)$ for all $j' \geq j$.
We can then again use the assumption that $\rho'$ is accepting.
Consequently, there is a $j'> j$ such that $(s_{j' - 1},\alpha_{j' - 1},s_{j'})$ is accepting, and $v_m^{j'} = (v,r)$, as the marker is circulating in a round robin fashion.
(Note that $q_{j'} \in l_r^{j'}(v)$ implies $l_r^{j'}(v) \neq \emptyset$.)
But then we have again $q_{j'} \in Q_m^{j'}= l_r^{j'}(v)$, and an inductive argument provides $(s_k,\alpha_k,s_{k+1})\notin F$ and $q_k \in Q_m^k$ for all $k\geq j'$. 
\smallskip

\noindent{\bf Contradiction.}\hspace{1ex}
As the level is reduced by two every second step, one of the arguments that contradict the assumption that $\rho'$ is accepting is reached in at most $\pi_e$ steps.
\end{proof}

\begin{corollary}
\label{cor:correct}
$\mathcal C$ recognises the complement language of $\mathcal P$.
\qed
\end{corollary}

\subsection{Lower bound and tightness}
\label{ssec:lower}

In order to establish a lower bound, we use a sub-language of the full automaton $\mathcal P_n^\Pi$, and show that an automaton that recognises it must have at least as many states as there are full FNHTs in $\fnht(Q,\pi)$ for $n=|Q|$ and $\pi = \max \Pi$.

To show this, we define two letters for each full FNHT $t =(\T,l_s,l_l,l_p,l_r) \in \fnht(Q,\pi)$.
$\beta_t: Q \times Q \rightarrow 2^\Pi$ is the letter where:
\begin{itemize}
 \item if $v$ is a stepchild and $p,q \in l_s(v)$, then $l_l(v){+}1 \in \beta_t(p,q)$ (provided $l_l(v){+}1 {\,\in\,} \Pi$),
 \item if $v$ is a stepchild, $p \in l_r(v)$, and $q \in l_s(vc)$ for some $c \in \omega$, then $l_l(v) \in \beta_t(p,q)$,
 \item if $v$ is a stepchild, $c,c' \in \omega$, $c<c'$, $vc' \in \T$, $p \in l_s(vc')$, and $q \in l_s(vc)$, then $l_l(v) \in \beta_t(p,q)$,
 \item if $v$ is a natural child, $p \in l_p(v)$, and $q \in l_r(v)$ then $l_l(v) \in \beta_t(p,q)$.
 \item if $v$ is a natural child and $p,q \in l_r(v)$, then $l_l(v)-1 \in \beta_t(p,q)$, and
 \item if $v$ is a natural child and $p,q \in l_p(v)$, then $l_l(v)-1 \in \beta_t(p,q)$.
\end{itemize}

\noindent 
$\gamma_t: Q \times Q \rightarrow 2^\Pi$ is the letter where $i \in \gamma_t(p,q)$ if $i \in \beta_t(p,q)$ and additionally:
\begin{itemize}
 \item if $v$ is a natural child, $l_l(v)-2 \in \Pi$, and $p,q \in l_r(v)$, then $l_l(v)-2 \in \gamma_t(p,q)$,
 \item if $v$ is a stepchild and $p,q \in l_r(v)$, then $l_l(v) \in \gamma_t(p,q)$, and
 \item if $v$ is a natural child, $l_l(v)-2 \in \Pi$, and $p,q \in l_p(v)$, then $l_l(v)-2 \in \gamma_t(p,q)$.
\end{itemize}

For a high integer $h > |\fnht(Q,\pi)|$, we now define the $\omega$-word
$\alpha^t = (\beta_t {\gamma_t}^{h-1})^\omega$, which consists of infinitely many sequences of length $h$ that start with a letter $\beta_t$ and continue with $h-1$ repetitions of the letter $\gamma_t$, for each full FNHT $t \in \fnht(Q,\pi)$.

We first observe that $\alpha^t$ is rejected by $\mathcal P_n^\Pi$.

\begin{lemma}
\label{lem:language}
$\alpha^t \notin \mathcal L (\mathcal P_n^\Pi)$. 
\end{lemma}

\begin{proof}
By Lemma~\ref{cor:correct}, it suffices to show that the complement automaton $\mathcal C$ of $\mathcal P_n^\Pi$, as defined in Section \ref{ssec:complement} accepts $\alpha^t$.
The language is constructed such that $\mathcal C$ has a run $\rho = Q (t;v_m^1,Q_m^1) (t;v_m^2,Q_m^2) (t;v_m^3,Q_m^3) \ldots$, such that the transition
$\big((t;v_m^i,Q_m^i),\alpha^t_i,(t;v_m^{i+1},Q_m^{i+1})\big)$ is accepting for $i > 0$ if $i \mod h =0$.
\end{proof}

Let $\mathcal B$ be some automaton with states $S$ that recognises the complement language of $\mathcal P_n^\Pi$.
We now fix an accepting run $\rho_t=s_0 s_1 s_2 \ldots$ for each word $\alpha^t$ and define the set $A_t$ of states in an `accepting cycle' as $A_t = \big\{s \in S \mid \exists i,j,k \in \omega \mbox{ with } 1\leq j < k \leq h \mbox{ such that } s = s_{ih+j} = s_{ih+k}\big\}$ holds, and define the interesting states $I_t = A_t \cap$ $\infi(\rho_t)$.

\begin{lemma}
\label{lem:disjoint}
For $t \neq t'$, $I_t$ and $I_{t'}$ are disjoint ($I_t \cap I_{t'} = \emptyset$).
\end{lemma}

\begin{proofidea}
 The proof idea is to assume that a state $s \in I_t \cap I_{t'}$, and use it to construct a word from $\alpha_t$ and $\alpha_{t'}$ and an accepting run of $\mathcal B$ on the resulting word from $\rho_t$ and $\rho_{t'}$, and then show that it is also accepted by $\mathcal P_n^\Pi$.
\end{proofidea}

\begin{proof}
Let us assume for contradiction that $s \in I_t \cap I_{t'}$ for 
$t = (\T,l_s,l_l,l_p,l_r) \neq t' = (\T',l_s',l_l',l_p',l_r')$.

Noting that we can change the role of $t$ and $t'$, we fix two positions $i$ and $i'$ in the run $\rho_t$ of $\alpha_t$ such that
$s = s_i = s_{i'}$, and there is a $j \in \omega$ such that $jh < i < i' \leq j(h+1)$, and two positions $j$ and $j'$ in $\rho_{t'} = s_0's_1's_2' \ldots$ such that $j<j'$, $s = s_j' = s_{j'}'$ and there is a $k\in \omega$ with $j \leq k < j'$ such that $(s_k',\alpha_k^{t'},s_{k+1}')$ is an accepting transition of $\mathcal B$.
Note that the definition of $I_t$ provides the first and the definition of $I_{t'}$ the latter.

For the finite words $\beta_1 = \alpha_0^t \alpha_1^t \ldots \alpha_{i-1}^t$, $\gamma_1=s_0 s_1 \ldots s_{i-1}$, $\beta_2 = {\gamma_t}^{i'-1}$, $\gamma_2=s_i s_{i+1} \ldots s_{i'-1}$, $\beta_3 = \alpha_j^{t'} \alpha_{j+1}^{t'} \ldots \alpha_{j'-1}^{t'}$, and $\gamma_3=s_j s_{j+1} \ldots s_{j'-1}$,
$\rho^{t'}_t = \gamma_1 (\gamma_2 \gamma_3)^\omega$ is an accepting run of the input word $\alpha^{t'}_t = \beta_1 (\beta_2 \beta_3)^\omega = \alpha_0 \alpha_1 \alpha_2 \ldots$.

We now show that $\alpha^{t'}_t$ or $\alpha^t_{t'}$ is accepted by $\mathcal P_n^\Pi$.

We start with the degenerated case that $\T = \{\varepsilon\}$ is the FNHT where the root is a leaf, and thus $\pi = \max \Pi$ odd.
(The case $\T' = \{\varepsilon\}$ is similar.)
We select a $q \in l_s'(0)$, and consider the run
$\rho=q^\omega$ of $\mathcal P_n^\Pi$ on $\alpha^{t'}_t$.
By construction, $\pri(q,\alpha_k,q) \leq \opt \Pi = l_l(\varepsilon)$ holds for all $k \geq i$.
Moreover, $\alpha_k = \gamma_t$ holds for infinitely many $k \in \omega$.
(In particular, it holds if $k\geq i$ and $(k - i) \mod (i'-i + j' - j) < i' - i$.)
For all of these transitions, $\pri(q,\alpha_k,q) = \opt \Pi= l_l(\varepsilon)$ holds, such that $\limsup_{n \rightarrow \infty}\big(\overline{\rho}(i)\big) = \opt \Pi$ is even.

Starting with the level $\lambda = \opt \Pi$ of the root and the whole trees $\T$ and $\T'$, we now run through the following construction.

We \textbf{firstly} look at the \textbf{case} that there is \textbf{some difference in the label of some natural child} $v \in \T \cap \T'$ on the level $\lambda$.
If there is an oldest child $v \in \T \cap \T'$ with $l_s(v) \neq l_s'(v)$, we assume w.l.o.g.\ that there is a $q \in l_s(v) \smallsetminus l_s'(v)$.
Then there are two sub-cases, first that there is a $q' \in l_s(v) \cap l_s'(v)$, and second that $l_s(v) \cap l_s'(v)= \emptyset$.
In the latter case we choose a $q' \in l_s'(v)$.
In both sub-cases, the run $\rho = q^{i'} (q'^{j'-j} q^{i'-i})^\omega = q_0 q_1 q_2 \ldots$ of $\mathcal P_n^\Pi$ on $\alpha^{t'}_t$ satisfies $\pri(q_k,\alpha_k,q_{k+1}) \succcurlyeq \lambda-1$ for all $k \in \omega$, and
$\pri(q_k,\alpha_k,q_{k+1}) \succcurlyeq \lambda$ when $q_k = q$ and $q_{k+1} = q'$. (Note that in this case $\alpha_k \in \{\beta_{t'},\gamma_{t'}\}$ holds.)

Otherwise $l_s(v) = l_s'(v)$ holds for all natural children $v \in \T \cap \T'$ on level $\lambda$, and there is a $v \in \T \cap \T'$ on level $\lambda$ such that $l_r(v) \neq l_r'(v)$.
We assume w.l.o.g.\ that there is a $q \in l_r(v)\smallsetminus l_r'(v)$.
We choose a $q' \in l_p(v)$. (Note that $q \neq q' \in l_s(v) = l_s'(v)$.)
Then the run $\rho = q^{i'} (q'^{j'-j} q^{i'-i})^\omega = q_0 q_1 q_2 \ldots$ of $\mathcal P_n^\Pi$ on $\alpha^{t'}_t$ satisfies $\pri(q_k,\alpha_k,q_{k+1}) \succcurlyeq \lambda-1$ for all $k \in \omega$, and
$\pri(q_k,\alpha_k,q_{k+1}) \succcurlyeq \lambda$ when $q_k = q$ and $q_{k+1} = q'$. (Note that in this case $\alpha_k = \gamma_t$ holds.)

We \textbf{secondly} look at the \textbf{case} where $l_s(v) = l_s'(v)$ and $l_r(v) = l_r'(v)$ holds for all natural children $v \in \T \cap \T$ on level $\pi_e$, but there is a natural \textbf{child} $v$ on level $\lambda$ \textbf{in the symmetrical difference} of $\T$ and $\T'$.
Let us assume w.l.o.g.\ that $v \in \T'$.
Let $q \in l_s'(v)$ and let $v$ be the child of $v'$.
This immediately implies that $q \in l_r(v)$.
Thus, the run $\rho = q^\omega$ of $\mathcal P_n^\Pi$ on $\alpha^{t'}_t$ satisfies $\pri(q,\alpha_k,q) \succcurlyeq \lambda - 1$ for all $k > i$, and $\pri(q,\alpha_k,q) \succcurlyeq \lambda$ whenever $\alpha_k = \gamma_t$, which happens infinitely often.

We \textbf{finally} look at the \textbf{case} where the nodes of $\T$ and $\T'$ on level $\lambda$ are the same, and where $l_s(v) = l_s'(v)$ and $l_r(v) = l_r'(v)$ hold for all nodes $v$ of $\T$ on level $\lambda$, but there is a node $v$ on level $\lambda$ which is a \textbf{leaf} in $\T$ but not in $\T'$. (The case ``leaf in $\T'$ but not in $\T$'' is entirely symmetric.)
Thus, $v\s 0$ is a node in $\T'$, and we select a $q \in l_s(v\s 0)$.
If we now consider the run $\rho = q^\omega$ of $\mathcal P_n^\Pi$ on $\alpha^{t'}_t$, then $\pri(q,\alpha_k,q) \succcurlyeq \lambda - 3$ holds for all $k > i$.
At the same time $\pri(q,\alpha_k,q) \succcurlyeq \lambda-2$ holds whenever $\alpha_k = \gamma_t$, which happens infinitely often.

If neither of these cases holds, then there must be a natural child $v$ on level $\lambda$ such that $v\s \in \T \cap \T'$ and $l_s(v\s) = l_p(v) = l_p'(v) = l_s'(v\s)$, such that $t$ and $t'$ differ on the descendants of $v$.
We then continue the construction by reducing $\lambda$ to $\lambda - 2$ and intersecting $\T$ and $\T'$ with the descendants of $v$ in $t$ and $t'$, respectively, and restrict the co-domain of the labelling functions of $t$ and $t'$ accordingly.
This construction will lead to a difference in at most $0.5 \cdot \opt \Pi$ steps.
\end{proof}

\begin{theorem}
\label{theo:lower}
$\mathcal B$ has at least as many states as $\fnht(Q,\max \Pi)$ contains full FNHTs.
\end{theorem}

\begin{proof}
We prove the claim with a case distinction.
The first case is that $I_t \neq \emptyset$ holds for all full FNHT $t \in \fnht(Q,\max \Pi)$.
Lemma \ref{lem:disjoint} shows that the sets of interesting states are pairwise disjoint for different trees $t \neq t'$, such that, as none of them is empty, $\mathcal B$ has at least as many states as $\fnht(Q,\max \Pi)$ contains full FNHTs.

The second case is there is a full FNHT $t \in \fnht(Q,\max \Pi)$ such that $I_t = \emptyset$.
By Lemma \ref{lem:language}, each $\rho_t=s_0 s_1 s_2 \ldots$ is an accepting run.
Let now $i \in \omega$ be an index, such that, for all $j \geq i$, $s_j \in \infi(\rho_t)$, and $k \geq i$ an integer with $k \mod h = 0$.
$I_t = \emptyset$ implies that $s_{k+j} \neq s_{k+j'}$ for all $j,j'$ with $1 \leq j<j' \leq h$.
Then $\mathcal B$, and even $\infi(\rho_t)$, has at least $h-1$ different states, and the claim follows with $h > |\fnht(Q,\max \Pi)|$.
\end{proof}

\label{ssec:tight}
To show tightness, we proceed in three steps.
In a first step, we provide an injection from MFTs with non-full marking to MFTs with full marking.

Next, we argue that the majority of FNHTs is full.
Taking into account that there are at most $|Q|$ different markers makes it simple to infer that the states of our complementation construction divided by the lower bound from Theorem \ref{theo:lower} is in $O(n)$.

\begin{lemma}
\label{lem:emptyMFT}
There is an injection from MFTs with non-full marking to MFTs with full marking in $\mft(Q,\pi)$.
\end{lemma}

\begin{proof}
For non-trivial trees $\T \neq \{\emptyset\}$, we can simply map an MFT $(\T,l_s,l_l,l_p,l_r;v_m,Q_m)$
\begin{itemize}
 \item for $v_m = (\overline{v},p)$ to the MFT $\big(\T',l_s',l_l',l_p',l_r';v_m,l_p'(\overline{v})\big)$ and
 \item for $v_m = (\overline{v},r)$ to the MFT $\big(\T',l_s',l_l',l_p',l_r';v_m,l_r'(\overline{v})\big)$, where
\end{itemize}
$\T'$ differs from $\T$ only in that it has a fresh node $v$, which is the youngest sibling of $v_m$.
$l_s',l_p',l_r'$ differ from $l_s,l_p,l_r$ only in $\overline{v}$ and $v$ (where $v$ is only in the pre-image of $l_s',l_l',l_p',l_r'$).
We set $l_s'(v) = l_p'(v) = Q_m$ and, consequently, $l_r'(v) = \emptyset$.
We also set $l_s'(v) = l_s(v) \smallsetminus Q_m$.

For $v_m = (\overline{v},p)$, we set
$l_r'(\overline{v}) = l_r(\overline{v})$ and $l_p'(\overline{v}) = l_p(\overline{v}) \smallsetminus Q_m$.
Note that, by the definition of markers, $\overline{v}$ is a leaf, and $l_p'(\overline{v})$ is non-empty because the marking in $(\T,l_s,l_l,l_p,l_r;v_m,Q_m)$ is not full.

For $v_m = (\overline{v},r)$, we set
$l_r'(\overline{v}) = l_r(\overline{v}) \smallsetminus Q_m$ and $l_p'(\overline{v}) = l_p(\overline{v})$.
Note that $l_r'(\overline{v})$ is non-empty because the marking in $(\T,l_s,l_l,l_p,l_r;v_m,Q_m)$ is not full.

It is easy to see that the resulting MFT is well formed in both cases.
What remains is the corner case of $\T = \{\varepsilon\}$.

$(\T,l_s,l_l,l_p,l_r;(\varepsilon,r),Q_m)$ and map it to $(\T',l_s',l_l',l_p',l_r';(\varepsilon,r),Q_m)$ for $\T' = \{\varepsilon,0\}$ and
$l_s'(\varepsilon)=l_s(\varepsilon)$, $l_s'(\varepsilon)=Q_m$, $l_p'(\varepsilon)=l_p'(0)=\emptyset$, and consequently $l_s'(0)= l_r'(0)= l_s(\varepsilon) \smallsetminus Q_m$.
(Note that the latter is non-empty because the marking in $(\T,l_s,l_l,l_p,l_r;(\varepsilon,r),Q_m)$ is not full.)
This is again a well formed MFT with full marking.

It is easy to see that the resulting function is injective.
\end{proof}

In Lemma \ref{lem:emptyMFT}, we have shown that the majority of MFTs have a full marking.
Next we will see that the majority of FNHTs is full. (Note that neither mapping is surjective.)

\begin{lemma}
\label{lem:fullFNHT}
There is an injection from non-full to full FNHTs in $\fnht(Q,\pi)$.
\end{lemma}

\begin{proof}
To obtain such an injection, it suffices to map a non-full FNHT $(\T,l_s',l_l',l_p',l_r')$ to the FNHT $(\T',l_s',l_l',l_p',l_r')$ where $\T'$ differs from $\T$ only in that it has a fresh youngest child $v$ of the root.

$l_s'$ agrees with $l_s$ on every node of $\T$ except for the root $\varepsilon$, and $l_p',l_r'$ agree with $l_p,l_r$ on every node of $\T$.
We set
$l_s' (\varepsilon) = Q$, $l_s' (v) = l_p'(v)=Q \smallsetminus l_s(\varepsilon)$, and $l_r'(v)= \emptyset$.

It is obvious that the resulting FNHT $(\T',l_s',l_l',l_p',l_r')$ is full and well formed, and it is also obvious that the mapping is injective.
\end{proof}

\begin{theorem}
The complementation construction from this section is tight up to a factor of $4n+1$, where $n=|Q|$ is the number of states of the complemented parity automaton.
\end{theorem}

\begin{proof}
For the number of MFTs, Lemma \ref{lem:emptyMFT} shows that they are at most twice the number of MFTs with full marking.
Note that the marker $(v_m,p)$ can only refer to leafs where $l_p(v_m)$ is non-empty and markers $(v_m,r)$ can only refer to nodes where $l_r(v_m)$ is non-empty.
It is easy to see that all sets described in this way are pairwise disjoint.
This implies that there are at most $|Q|$ such markers.
Thus, the number of MFTs with full marking is at most $n$ times the number of FNHTs.

By Lemma \ref{lem:fullFNHT}, the number of FNHTs is in turn at most twice as high as the number of all full FNHTs.
Thus we have bounded the number of MFTs by $4n$ times the number of full FNHTs used to estimate the lower bound in Theorem \ref{theo:lower}, irrespective of the priorities.

What remains is the trivial observation that the second part of the state-space, the subset construction, is dwarfed by the number of MFTs.
%For simplicity, we estimate the number of full FNHTs by $2^n$.
Consequently, we can estimate the state-space of the complement automaton divided by the lower bound from Theorem \ref{theo:lower} by $4n + 1$.
\end{proof}

\end{document}